\documentclass[a4paper]{article}
\usepackage{microtype}

\usepackage{xspace,url}
\usepackage{array,paralist}
\usepackage{color}
\usepackage{amsmath,amssymb,latexsym}

\usepackage{amsmath}
\usepackage{amssymb}
\usepackage{amsthm}
\newtheorem{lemma}{Lemma}

\newtheorem{example}{Example}
\newtheorem{definition}{Definition}
\newtheorem{theorem}{Theorem}
\newtheorem{proposition}{Proposition}

\newcommand{\set}[1]{\left\{#1\right\}}

\newcommand{\simmed}[2]{#1^{#2}}

\newcommand{\etal}{\emph{et al.}\xspace}

\usepackage[ruled,lined,linesnumbered,boxed,vlined]{algorithm2e}

\title{On the Workflow Satisfiability Problem with Class-Independent Constraints}

\author{Jason Crampton,
Andrei Gagarin, 
Gregory Gutin,
Mark Jones, \and
Magnus Wahlstr\"om\\
{\footnotesize Royal Holloway, University of London, Egham, Surrey, TW20 0EX, UK}
}

\begin{document}
\maketitle

\begin{abstract}
A workflow specification defines sets of steps and users.
An authorization policy determines for each user a subset of steps the user is allowed to perform.  
Other security requirements, such as separation-of-duty, impose constraints on which subsets of users may perform certain subsets of steps.
The \emph{workflow satisfiability problem} (WSP) is the problem of determining whether there exists an assignment of users to workflow steps that satisfies all such authorizations and constraints.
An algorithm for solving WSP is important, both as a static analysis tool for workflow specifications, and for the construction of run-time reference monitors for workflow management systems.
Given the computational difficulty of WSP, it is important, particularly for the second application, that such algorithms are as efficient as possible.

We introduce class-independent constraints, enabling us to model scenarios where the set of users is partitioned into groups, and the identities of the user groups are irrelevant to the satisfaction of the constraint. We prove that solving WSP is fixed-parameter tractable (FPT) for this class of constraints and develop an FPT algorithm that is useful in practice. We compare the performance of the FPT algorithm with that of SAT4J (a pseudo-Boolean SAT solver) in computational experiments, which show that our algorithm significantly outperforms SAT4J for many instances of WSP. User-independent constraints, a large class of constraints including many practical ones, are a special case of class-independent constraints for which WSP was proved to be FPT (Cohen {\em et al.}, J. Artif. Intel. Res. 2014). Thus our results considerably extend our knowledge of the fixed-parameter tractability of WSP.\end{abstract}

\section{Introduction}

It is increasingly common for organizations to computerize their business and management processes.
The co-ordination of the tasks or steps that comprise a computerized business process is managed by a workflow management system (or business process management system).
Typically, the execution of these steps will be triggered by a human user, or a software agent acting under the control of a human user, and each step may only be executed by an \emph{authorized} user.
Thus a workflow specification will include an authorization policy defining which users are authorized to perform which steps.

In addition, many workflows require controls on the users that perform certain sets of steps~\cite{ansi-rbac04,BaBuKa14,BrNa89,Cr05,WaLi10}.
%
Consider a simple purchase-order system in which there are four steps: raise-order ($s_1$), acknowledge-receipt-of-goods ($s_2$), raise-invoice ($s_3$), and send-payment ($s_4$).  
%
The workflow specification for the purchase-order system includes rules to prevent fraudulent use of the system, the rules taking the form of \emph{constraints} on users that can perform pairs of steps in the workflow: 
the same user may not raise the invoice ($s_3$) and sign for the goods ($s_2$), for example.
Such a constraint is known as a \emph{user-independent} (UI) constraint, since the specific identities of the users that perform these steps are not important, only the relationship between them (in this example, the identities must be different).

Once we introduce constraints on the execution of workflow steps, it may be impossible to find a \emph{valid plan}~--~an assignment of authorized users to workflow steps such that all constraints are satisfied.
The {\sc Workflow Satisfiability Problem} (WSP) takes a workflow specification as input and outputs a valid plan if one exists. 
WSP is known to be NP-hard, even when the set of constraints only includes constraints having a relatively simple structure (and arising regularly in practice).
In particular, the {\sc Graph $k$-Colorability} problem can be reduced to a special case of WSP in which the workflow specification only includes separation-of-duty constraints~\cite{WaLi10}.
Clearly, it is important to be able to determine whether a workflow specification is satisfiable at design time.
Equally, when users select steps to execute in a workflow instance, it is essential that the access control mechanism can determine whether
\begin{inparaenum}[(a)]
 \item the user is authorized,
 \item allowing the user to execute the step would render the instance unsatisfiable. 
\end{inparaenum}
Thus, the access control mechanism must incorporate an algorithm to solve WSP, and that algorithm needs to be as efficient as possible.

Wang and Li~\cite{WaLi10} observed that, in practice, the number $k$ of steps in a workflow will be small, relative to the size of the input to WSP; specifically, the number of users is likely to be an order of magnitude greater than the number of steps.
This observation led them to set $k$ as the parameter and to study the problem using tools from parameterized complexity. 
In doing so, they proved that the problem is {\em fixed-parameter tractable} (FPT) for simple classes of constraints.
%
However, Wang and Li also showed that for many types of constraints the problem is fixed-parameter \emph{intractable} (unless $\text{FPT}\neq\text{W[1]}$ is false).
Hence, it is important to be able to identify those types of practical constraints for which WSP is FPT.

Recent research has made significant progress in understanding the fixed-parameter tractability of WSP.
In particular, Cohen \etal~\cite{CoCrGaGuJo14}  introduced the notion of patterns and, using it, proved that WSP is FPT (irrespective of the authorization policy) if all constraints in the specification are UI. This result is significant because most constraints in the literature~--~including separation-of-duty, cardinality and counting constraints~--~are UI~\cite{CoCrGaGuJo14}. Using a modified pattern approach, Karapetyan \etal~\cite{KaGaGu} provided both a short proof that WSP with only UI constraints is FPT and a very efficient algorithm for WSP with UI constraints. 

However, it is known that not all constraints that may be useful in practice are UI. 
Consider a situation where the set of users is partitioned into groups (such as departments or teams) and we wish to define constraints on the groups, rather than users.
In our purchase order example, suppose each user belongs to a specific department.
Then it would be reasonable to require that steps $s_1$ and $s_2$ are performed by different users belonging to the same department.
There is little work in the literature on constraints of this form, although prior work has recognized that such constraints are likely to be important in practice~\cite{Cr05,WaLi10}, and 
it has been shown 
that such constraints present additional difficulties when incorporated into WSP~\cite{CrGuYe13}.

In this paper, we extend the notion of a UI constraint to that of a \emph{class-independent} (CI) constraint.
In particular, every UI constraint is an instance of a CI constraint.
Our second contribution is to demonstrate that patterns for UI constraints~\cite{CoCrGaGuJo14} can be generalized to patterns for CI constraints, as well as to ``nested'' CI constraints in several levels. 
The resulting algorithm, using these new patterns, remains FPT (irrespective of the authorization policy), although its running time is somewhat slower than that of the algorithm for WSP with UI constraints only.
In short, our first two contributions identify a large class of constraints for which WSP is shown to be FPT, and subsume prior work in this area~\cite{CrGuYe13,CoCrGaGuJo14,WaLi10}.
Our final contribution is an implementation of our algorithm in order to investigate whether the theoretical advantages implied by its fixed-parameter tractability can be realized in practice. We compare our FPT algorithm with SAT4J, an off-the-shelf pseudo-Boolean (PB) SAT solver.
The results of our experiments suggest that our FPT algorithm enjoys some significant advantages over SAT4J for hard instances of WSP.

In the next section, we define WSP and UI constraints in more formal terms, discuss related work in more detail, and introduce the notion of class-independent constraints.
In Sections~\ref{sec:pp} and \ref{sec:real}, we state and prove a number of technical results that underpin the algorithm for solving WSP with class-independent constraints. We describe the algorithm and establish its worst-case complexity in Section \ref{sec:main}.  In Section~\ref{sec:nested}, we describe the generalisations to several levels of nested CI constraints, and analyse the resulting running time more carefully.  In Section~\ref{sec:experiments}, we describe our experimental methods and report the results of our experiments.  We conclude in Section \ref{sec:con}.

In the main part of the paper, we focus on the case of a single non-trivial partition of the user set. 
The treatment of the case with nested CI constraints -- i.e., multiple nested partitions of the user set -- is confined to Section~\ref{sec:nested}. 
(Nested CI constraints can be used to model hierarchical organizational structures, which can be useful in practice~\cite{CrGuYe13}.)

\section{Workflow Satisfiability}

Let $S = \{s_1, \dots, s_k\}$ be a set of \emph{steps}, let $U = \{u_1, \dots, u_n\}$ be a set of \emph{users} in a workflow specification, and let $k\le n.$
We are interested in assigning users to steps subject to certain constraints.
In other words, among the set $\Pi(S,U)$ of 
functions from $S$ to $U$, there are some that represent ``legitimate'' assignments of steps to users 
and some that do not.

The legitimacy or otherwise of an assignment is determined by the authorization policy and the constraints that complete the workflow specification.
Let ${\cal A} = \{A(u): u \in U\}$ be a set of \emph{authorization lists}, where $A(u)\subseteq S$ for each $u\in U$, and let $C$ be a set of \emph{(workflow) constraints}.
A \emph{constraint} $c \in C$ may be viewed as a pair $(T, \Theta)$, where $T \subseteq S$ is the \emph{scope} of $c$ and $\Theta$ is a set of functions from $T$ to $U$, specifying the assignments of steps in $T$ to users in $U$ that satisfy the constraint.
In practice, we do not enumerate all the elements of $\Theta$.  Instead, we define its members implicitly using some constraint-specific syntax.  In particular, we write $(s,s',\rho)$, where $s,s' \in S$ and $\rho$ is a binary relation defined on $U$, to denote a constraint that has scope $\set{s,s'}$ and is satisfied by any plan $\pi : S \rightarrow U$ such that $(\pi(s),\pi(s')) \in \rho$.
Thus $(s,s',\ne)$, for example, requires $s$ and $s'$ to be performed by different users (and so represents a separation-of-duty constraint). Also $(s,s',=)$ states that $s$ and
$s'$ must be performed by the same user (a binding-of-duty constraint).

\subsection{The Workflow Satisfiability Problem}

A \emph{plan} is a function in $\Pi(S,U)$.
Given a \emph{workflow} $W = (S,U, {\cal A}, C)$, a plan $\pi$ is \emph{authorized} if for all $s \in S$, $s \in A(\pi(s))$, i.e. the user assigned to $s$ is authorized for $s$.
A plan $\pi$ is \emph{eligible} if for all $(T, \Theta) \in C$, $\pi|_{T}  \in \Theta$, i.e. every constraint is satisfied.
A plan $\pi$ is \emph{valid} if it is both authorized and eligible. 
In the \emph{workflow satisfiability problem (WSP)}, we are given a workflow (specification) $W$, and our aim is to decide whether $W$ has a valid plan.
If $W$ has a valid plan, $W$ is \emph{satisfiable}; otherwise, $W$ is \emph{unsatisfiable}.

Note that WSP is, in fact, the conservative CSP (i.e., CSP with unary constraints corresponding to step authorizations in the WSP terminology). However, unlike a typical instance of CSP, where the number of variables is significantly larger than the number of values, a typical instance of WSP has many more values (i.e., users) than variables (i.e., steps). 

We assume that in all instances of WSP we consider, all constraints can be checked in time polynomial in $n$.
Thus it takes polynomial time to check whether any plan is  eligible. The correctness of our algorithm is unaffected by this assumption, but using constraints not checkable in polynomial time would naturally affect the running time.

\begin{example}\label{ex1} {\rm Consider the following instance $W'$ of WSP\@.
The step and user sets are $S = \{ s_1, s_2, s_3, s_4 \}$ and $U = \{u_1,u_2,u_3,u_4, u_5\}$. 
The authorization lists are $A(u_1) = \{ s_1, s_2, s_3, s_4 \}$, $A(u_2) = \{ s_1 \}$, $A(u_3) = \{ s_2 \}$, $A(u_4) = A(u_5) = \{ s_3, s_4 \}$. 
The constraints are $(s_1,s_2,=)$, $(s_2, s_3,\neq)$,  $(s_3, s_4,\neq)$, and $(s_4,s_1,\neq)$. Observe that $\pi':\ S \rightarrow U$ with $\pi'(s_1)=\pi'(s_2)=u_1$, $\pi'(s_3)=u_5$ and $\pi'(s_4)=u_4$ satisfies all constraints and authorizations, and thus $\pi'$ is a valid plan for $W'$. Therefore, $W'$ is satisfiable.}
\end{example}

\subsection{Constraints using Equivalence Relations}

Crampton \etal~\cite{CrGuYe13} introduced constraints defined in terms of an equivalence relation $\sim$ on $U$: a plan $\pi$ satisfies constraint $(s,s',\sim)$ if $\pi(s) \sim \pi(s')$  (and satisfies constraint $(s,s',\nsim)$ if $\pi(s) \nsim \pi(s')$).
Hence, we could, for example, specify the pair of constraints $(s,s',\ne)$ and $(s,s',\sim)$, which, collectively, require that $s$ and $s'$ are performed by different users that belong to the same equivalence class.
As we noted in the introduction, such constraints are very natural in the context of organizations that partition the set of users into departments, groups or teams.

Moreover, Crampton \etal~\cite{CrGuYe13} demonstrated that ``nested'' equivalence relations can be used to model hierarchical structures within an organization\footnote{Many
    organizations exhibit nested hierarchical structure. For example, the
    academic parts of many universities are divided into faculties/schools
    which are divided into departments.}
 and to define constraints on workflow execution with respect to those structures.
More formally, an equivalence relation $\sim$ is said to be a \emph{refinement} of an equivalence relation $\approx$ if $x \sim y$ implies $x \approx y$.
In particular, given an equivalence relation $\sim$, $=$ is a refinement of $\sim$.
Crampton \etal proved that WSP remains FPT when some simple extensions of constraints $(s,s',\sim)$ and $(s,s',\nsim)$ are included~\cite[Theorem 5.4]{CrGuYe13}.
Our extension of constraints $(s,s',\sim)$ and $(s,s',\nsim)$ is much more general: it is similar to generalizing simple constraints $(s,s',=)$ and $(s,s',\neq)$ to the wide class of UI constraints. This leads, in particular, to a significant generalization of Theorem 5.4 in \cite{CrGuYe13}.


 Let $c = (T,\Theta)$ be a constraint and let $\sim$ be an equivalence relation on $U$.
 Let $\simmed{U}{\sim}$ denote the set of equivalence classes induced by $\sim$ and let $\simmed{u}{\sim} \in \simmed{U}{\sim}$ denote the equivalence class containing $u$.
 Then, for any function $\pi : S \rightarrow U$, we may define the function $\simmed{\pi}{\sim} : S \rightarrow \simmed{U}{\sim}$, where $\simmed{\pi}{\sim}(s) = \simmed{(\pi(s))}{\sim}$.
 In particular, $\sim$ induces a set of functions $\simmed{\Theta}{\sim} = \set{\simmed{\theta}{\sim}: \theta \in \Theta}$.

\begin{example}\label{ex2} {\rm Continuing from Example \ref{ex1}, suppose $U^{\sim}$ consists of two equivalence classes $U_1=\{u_1,u_2,u_5\}$ and $U_2=\{u_3,u_4\}$. Let us add to $W'$ another constraint $(s_1,s_4,\sim)$ ($s_1$ and $s_4$ must be assigned users from the same equivalence class) to form a new instance $W''$ of WSP. Then plan $\pi'$ does not satisfy the added constraint and so $\pi'$ is not valid for $W''$. However, $\pi'':\ S \rightarrow U$ with $\pi''(s_1)=\pi''(s_2)=u_1$, $\pi''(s_3)=u_4$ and $\pi''(s_4)=u_5$ satisfies all constraints and authorizations, and thus $\pi''$ is valid for $W''$. Here $(\pi'')^{\sim}(s_1)=(\pi'')^{\sim}(s_2)=(\pi'')^{\sim}(s_4)=U_1$ and $(\pi'')^{\sim}(s_3)=U_2.$}
\end{example}
 
Given an equivalence relation $\sim$ on $U$, we say that a constraint $c = (T, \Theta)$ is \emph{class-independent (CI) for $\sim$} if 
$\simmed{\theta}{\sim} \in \simmed{\Theta}{\sim}$ implies $\theta \in \Theta$, and 
for any permutation $\phi: \simmed{U}{\sim} \rightarrow \simmed{U}{\sim}$, 
 $\simmed{\theta}{\sim} \in \simmed{\Theta}{\sim}$ implies $\phi \circ \simmed{\theta}{\sim} \in \simmed{\Theta}{\sim}$.
In other words, if a plan $\pi:\ s\mapsto \pi(s)$ satisfies a constraint $c$, which is class-independent for $\sim$, then for each permutation $\phi$ of classes in $\simmed{U}{\sim}$, if  we replace $\pi(s)$ by any user in the class $\phi(\pi(s)^{\sim})$ for every step $s$, then the new plan will satisfy $c$.

 We say a constraint is \emph{user-independent (UI)} if it is CI for $=$.
In other words, if a plan $\pi:\ s\mapsto \pi(s)$ satisfies a UI constraint $c$ and we replace any user in $\{\pi(s):\ s\in S\}$ by an arbitrary user such that the replacement users are all distinct, then the new plan satisfies $c$. 

We conclude this section with a  claim whose simple proof is omitted.

 \begin{proposition}\label{pro:plan-coherence}
 {\em  Given two equivalence relations $\sim$ and $\approx$ such that $\sim$ is a refinement of $\approx$, and any plan $\pi : S \rightarrow U$, $\simmed{\pi}{\sim}(s) = \simmed{\pi}{\sim}(s')$ implies $\simmed{\pi}{\approx}(s) = \simmed{\pi}{\approx}(s')$.}
 \end{proposition}

\section{Plans and Patterns}\label{sec:pp}

In what follows, unless specified otherwise, we will consider the equivalence relation $=$ along with another fixed equivalence relation $\sim$.
We will write $[m]$ to denote the set $\set{1,\dots,m}$ for any integer $m \geqslant 1$.
For brevity and simplicity of presentation, we assume for now that all constraints are either UI or CI for $\sim$ (i.e., we consider only two equivalence relations $=$ and $\sim$);
we will refer to constraints that are CI for $\sim$ as simply \emph{CI}.
In Section~\ref{sec:nested}, we generalise our results 
to any sequence $\sim_1, \dots, \sim_l$ of equivalence relations such that $\sim_{i+1}$ is a refinement of $\sim_i$ for all $i \in [l-1]$. 
It is important to keep in mind that we put {\em no restrictions on authorizations}.

We will represent groups of plans as \emph{patterns}.
The intuition is that a pattern defines a partition of the set of steps relevant to a set of constraints. 
For instance, suppose that we only have UI constraints.
Then a pattern specifies which sets of steps are to be assigned to the same user.
A pattern assigns an integer to each step and those steps that are labelled by the same integer will be mapped to the same user.
A pattern $p$ defines an equivalence relation $\sim_p$ on the set of steps (where $s \sim_p s'$ if and only if $s$ and $s'$ are assigned the same label).
Moreover, this pattern can be used to define a plan by mapping each of the equivalence classes induced by $\sim_p$ to a different user.
Since we only consider UI constraints, the identities of the users are irrelevant (provided they are distinct).
Conversely, any plan $\pi : S \rightarrow U$ defines a pattern: $s$ and $s'$ are labelled with the same integer if and only if $\pi(s) = \pi(s')$.
And if $\pi$ satisfies a UI constraint $c$, then any other plan with the same pattern will also satisfy $c$.
We can extend this notion of a pattern to CI constraints where entries in the pattern encode equivalence classes of users instead of single users.

More formally, let  $W = (S,U, {\cal A}, C = C_{=} \cup C_{\sim})$ be a workflow, 
where $C_{=}$ is a set of UI constraints and $C_{\sim}$ is a set of CI constraints.
Let $p_{=} = (x_1, \dots, x_k)$ where $x_i \in [k]$ for all $i \in [k]$.
We say that $p_{=}$ is a \emph{UI-pattern} for a plan $\pi$ if $x_i = x_j \Leftrightarrow \pi(s_i)=\pi(s_j)$, for all $i,j \in [k]$, and $p_{=}$ is \emph{eligible for $C_{=}$} if 
any plan $\pi$ with $p_{=}$ as its UI-pattern is eligible for $C_{=}$.

In Example \ref{ex2}, $C_{=}=\{(s_1,s_2,=),(s_2,s_3,\neq),(s_3,s_4,\neq), (s_1,s_4,\neq)\}$ and \mbox{$C_{\sim}=\{(s_1,s_4,\sim)\}.$}
Tuples $(1,1,2,3)$ and $(2,2,4,3)$ are UI-patterns for plan $\pi''$ of Example \ref{ex2}. 

\begin{proposition}\label{prop:ui}
{\em Let $p_{=}$ be a UI-pattern for a plan $\pi$. Then $p_{=}$ is eligible for $C_{=}$ if and only if $\pi$ is eligible for $C_{=}$. }
\end{proposition}
\begin{proof}
  Suppose that $\pi$ is eligible for $C_{=}$. We show that $\pi_0$ is eligible for $C_{=}$, for any plan $\pi_0$ that has $p_{=}$ as its UI-pattern, and so $p_{=}$ is eligible for $C_{=}$.
 
 Let $p_{=} = (x_1, \dots, x_k)$.
 Observe that for any $s_i,s_j$, we have $\pi(s_i) = \pi(s_j) \Leftrightarrow x_i = x_j \Leftrightarrow \pi_0(s_i) = \pi_0(s_j)$.
 Then define a permutation $\phi: U \rightarrow U$ as follows: $\phi(u) = \pi_0(s_i)$ if there exists $s_i \in S$ such that $\pi(s_i) = u$, and $\phi(u) = u$ otherwise. 
 As $\pi_0(s_i) = \pi_0(s_j)$ for any $s_i,s_j$ such that $\pi(s_i)=\pi(s_j)=u$, $\phi$ is well-defined.
 Furthermore $\pi_0 = \phi \circ \pi$. 
 Then it follows from the definition of a user-independent constraint that for any $c = (T,\Theta) \in C_{=}$, $\pi|_{T} \in \Theta \Leftrightarrow {\pi_0}|_{T} \in \Theta$.
 It follows that as $\pi$ satisfies every constraint in $C_{=}$, $\pi_0$ satisfies every constraint in $C_{=}$ and so $\pi_0$ is eligible for $C_{=}$, as required.
 
 For the converse, it follows by definition that if $p_{=}$ is eligible for $C_{=}$ then $\pi$ is eligible for $C_{=}$.
\end{proof}

Let $p_{\sim} = (y_1, \dots, y_k)$, where $y_i \in [k]$ for all $i \in [k]$.
We say that $p_{\sim}$ is a \emph{CI-pattern} for a plan $\pi$ if $y_i=y_j \Leftrightarrow \simmed{\pi}{\sim}(s_i) = \simmed{\pi}{\sim}(s_j)$, for all $i,j \in [k]$, and
$p_{\sim}$ is \emph{eligible for $C_{\sim}$} if 
any plan $\pi$ with $p_{\sim}$ as its CI-pattern is eligible for $C_{\sim}$.
For example, $(1,1,2,1)$ and $(2,2,4,2)$ are CI-patterns for plan $\pi''$ of Example \ref{ex2}. The next result is a generalization of Proposition \ref{prop:ui}.

\begin{proposition}\label{prop:ci}
{\em Let $p_{\sim}$ be a CI-pattern for a plan $\pi$. Then $p_{\sim}$ is eligible for $C_{\sim}$ if and only if $\pi$ is eligible for $C_{\sim}$. }
\end{proposition}
\begin{proof}
 Suppose that $\pi$ is eligible for $C_{\sim}$. We show that $\pi_0$ is eligible for $C_{\sim}$, for any plan $\pi_0$ that has $p_{\sim}$ as its CI-pattern, and so $p_{\sim}$ is eligible for $C_{\sim}$.
 
 Let $p_{\sim} = (y_1, \dots, y_k)$.
 Observe that for any $s_i,s_j$, we have $\simmed{\pi}{\sim}(s_i) =  \simmed{\pi}{\sim}(s_j) \Leftrightarrow y_i = y_j \Leftrightarrow \simmed{\pi_0}{\sim}(s_i) = \simmed{\pi_0}{\sim}(s_j)$.
  Then define a permutation $\phi: \simmed{U}{\sim} \rightarrow \simmed{U}{\sim}$ as follows: $\phi(\simmed{u}{\sim}) = \simmed{\pi_0}{\sim}(s_i)$ if there exists $s_i \in S$ such that $\simmed{\pi}{\sim}(s_i) = \simmed{u}{\sim}$, and $\phi(\simmed{u}{\sim}) = \simmed{u}{\sim}$ otherwise. 
 As $\simmed{\pi_0}{\sim}(s_i) = \simmed{\pi}{\sim}(s_j)$ for any $s_i,s_j$ such that $\simmed{\pi}{\sim}(s_i)=\simmed{\pi}{\sim}(s_j)=\simmed{u}{\sim}$, $\phi$ is well-defined.
 Furthermore $\simmed{\pi_0}{\sim} = \phi \circ \simmed{\pi}{\sim}$.

 Then it follows from the definition of a class-independent constraint that for any $c = (T,\Theta) \in C_{\sim}$, 
 $\pi|_{T} \in \Theta 
 \Leftrightarrow {\simmed{\pi}{\sim}}|_{T} \in \simmed{\Theta}{\sim}
 \Leftrightarrow \phi \circ ({\simmed{\pi}{\sim}}|_{T}) \in \simmed{\Theta}{\sim}
 \Leftrightarrow {\simmed{\pi_0}{\sim}}|_{T} \in \simmed{\Theta}{\sim} 
 \Leftrightarrow  {\pi_0}|_{T} \in \Theta$.
 It follows that as $\pi$ satisfies every constraint in $C_{\sim}$, $\pi_0$ satisfies every constraint in $C_{\sim}$ and so $\pi_0$ is eligible for $C_{\sim}$, as required.
 
 For the converse, it follows by definition that if $p_{\sim}$ is eligible for $C_{\sim}$ then $\pi$ is eligible for $C_{\sim}$.
  \end{proof}

Now let $p = (p_{=}, p_{\sim})$ be a pair containing a UI-pattern and an CI-pattern. Then we call $p$ a \emph{(UI, CI)-pattern}. We say that $p$ is a (UI, CI)-pattern for $\pi$ if $p_{=}$ is a UI-pattern for $\pi$ and $p_{\sim}$ is a CI-pattern for $\pi$. 
We say that $p$ is \emph{eligible for $C=C_{=} \cup  C_{\sim}$} if $p_{=}$ is eligible for $C_{=}$ and $p_{\sim}$ is eligible for $C_{\sim}$.
The following two results follow immediately from Propositions~\ref{prop:ui} and~\ref{prop:ci} and definitions of UI- and CI-patterns.

\begin{lemma}\label{lem:elig}
 Let $p = (p_{=}, p_{\sim})$ be a (UI, CI)-pattern for a plan $\pi$. 
 Then $p$ is eligible for $C = C_{=} \cup C_{\sim}$ if and only if $\pi$ is eligible for $C$.
\end{lemma}

\begin{proposition}\label{prop:4}
{\em There is a (UI, CI)-pattern $p$ for every plan $\pi$.}
\end{proposition}
 
We say a (UI, CI)-pattern $p$ is \emph{realizable} if there exists a plan $\pi$ such that $\pi$ is authorized and $p$ is a (UI, CI)-pattern for $\pi$.
Given the above results, in order to solve a WSP instance with user- and class-independent constraints, it is enough to decide whether there exists a (UI, CI)-pattern $p$ such that $\mbox{}$
\begin{inparaenum}[(i)]
 \item $p$ is realizable, and 
 \item $p$ is eligible (and hence $\pi$ is eligible) for $C=C_{=}\cup C_{\sim}$.
\end{inparaenum}

We will enumerate all possible (UI, CI)-patterns, and for each one check whether the two conditions hold. 
We defer the explanation of how to determine whether $p$ is realizable until Sec.~\ref{sec:real}.
We now show it is possible to check whether a (UI, CI)-pattern $p = (p_=,p_\sim)$ is eligible in time polynomial in the input size $N$. 
Indeed, in polynomial time, we can construct plans $\pi_=$ and $\pi_\sim$ with patterns $p_=$ and $p_\sim$, respectively, where $\pi_{=}(s_i) = \pi_{=}(s_j)$ if and only if $x_i = x_j$ and
$\pi_\sim(s_i) \sim \pi_\sim(s_j)$ if and only if $y_i = y_j$. (In particular, we can select a representative user from each equivalence class in $\simmed{U}{\sim}$.)
By Lemma \ref{lem:elig} and Propositions \ref{prop:ui} and \ref{prop:ci}, $p$ is eligible if and only if both $\pi_{=}$ and $\pi_{\sim}$ are eligible. 
By our assumption before Example \ref{ex1}, eligibility of both $\pi_{=}$ and $\pi_{\sim}$ can be checked in polynomial time.\footnote{Clearly, it is not hard to check eligibility of $p$ without explicitly constructing $\pi_{=}$ and $\pi_{\sim}$, as is done in our algorithm implementation, described in Section~\ref{sec:experiments}.\footnote{ORLY?}} 
Note, however, that $\pi_{=}$ and $\pi_{\sim}$ may be different plans, so this simple check for eligibility does not give us a check for realizability of $p$.

\section{Checking Realizability}\label{sec:real}



A \emph{partial plan} $\pi$ is a function from a subset $T$ of $S$ to $U$. In particular, a plan is a partial plan.  To avoid confusion with partial
plans, sometimes we will call plans {\em complete plans}.
 We can easily extend the definitions of \emph{eligible, authorized} and \emph{valid} plans to partial plans: the only difference is that we only consider authorizations for steps in $T$ and constraints with scope being a subset of $T$. 

We also define \emph{partial patterns}.
For a UI or CI-pattern $q = (x_1, \dots, x_k)$ and a subset $T \subseteq S$, let the pattern $q|_{T} = (z_1, \dots, z_k)$,
where $z_i = x_i$ if $s_i \in T$, and $z_i = 0$ otherwise.
We say that $p|_{T}$ is a (UI, CI)-pattern for a partial plan $\pi: T \rightarrow U$ if $p|_{T}$ with all coordinates with $0$ values removed is a (UI, CI)-pattern for $\pi$. 
We therefore have that if $p$ is a (UI, CI)-pattern for a plan $\pi$,
then $p|_{T}$ is a (UI, CI)-pattern for $\pi$ restricted to $T$.

Let $p = (p_{=} = (x_1, \dots, x_k), p_{\sim} = (y_1, \dots, y_k))$ be a (UI, CI)-pattern. We say that $p$ is \emph{consistent} if $x_i=x_j \Rightarrow y_i=y_j$ for all $i,j \in [k]$. 
Recall that if $p$ is the (UI, CI)-pattern for $\pi$, then $x_i=x_j \Leftrightarrow \pi(s_i) = \pi(s_j)$, and $y_i=y_j \Leftrightarrow \simmed{\pi}{\sim}(s_i) = \simmed{\pi}{\sim}(s_j)$.
Thus Proposition \ref{pro:plan-coherence} implies that if $p$ is the (UI, CI)-pattern for any plan then $p$ is consistent.
Henceforth, we will only consider (UI, CI)-patterns that are consistent.

Given a (UI, CI)-pattern $(p_=,p_{\sim})$, we must determine whether this (UI, CI)-pattern can be realized, given the authorization lists defined on users.
The patterns $p_{=}$ and $p_{\sim}$ define two sets of equivalence classes on $S$:
$s_i$ and $s_j$ are in the same equivalence class of $S$ defined by $p_{=}$ ($p_{\sim}$, respectively) if and only if $x_i=x_j$ ($y_i=y_j$, respectively).

Moreover each equivalence class induced by $p_{\sim}$ is partitioned by equivalence classes induced by $p_=$.
We must determine whether there exists a plan $\pi : S \rightarrow U$ that simultaneously $\mbox{}$%
\begin{inparaenum}[(i)]
 \item has UI-pattern $p_{=}$;
 \item has CI-pattern $p_{\sim}$; and
 \item assigns an authorized user to each step.
\end{inparaenum}
Informally, our algorithm for checking realizability computes two things.
\begin{itemize}
 \item For each pair $(T,V)$, where $T \subseteq S$ is an equivalence class induced by $p_\sim$ and $V \subseteq U$ is an equivalence class induced by $\sim$, whether there exists an injective mapping from the equivalence classes in $T$ induced by $p_=$ to authorized users in $V$.
 We call such a mapping a \emph{second-level} mapping.
 \item Whether there exists an injective mapping $f$ from the set of equivalence classes induced by $p_{\sim}$ to the set of equivalence classes induced by $\sim$ such that $f(T) = V$ only if there exists a second-level mapping from $T$ to $V$.  We call $f$ a \emph{top-level} mapping.
\end{itemize}
If a top-level mapping exists, then, by construction, it can be ``deconstructed'' into authorized partial plans defined by second-level mappings.
We compute top- and second-level mappings using matchings in bipartite graphs, as described below.\\

\noindent{\bf The Top-level Bipartite Graph.}
The UI-pattern $p_{=} = (x_1, \dots, x_k)$ induces an equivalence relation on $S = \{s_1, \dots, s_k\}$,
where $s_i$ and $s_j$ are equivalent if and only if $x_i = x_k$. 
Let ${\cal S} = \{S_1, \dots, S_l\}$ be the set of equivalence classes of $S$ under this relation.
Similarly, the CI-pattern $p_{\sim} = (y_1, \dots, y_k)$ induces an equivalence relation on $S$, where $s_i,s_j$ are equivalent if and only if $y_i = y_j$. Let ${\cal T} = \{T_1, \dots, T_m\}$ be the equivalence classes under this relation. 
Observe that since $p$ is consistent, we have $k \ge l \ge m$ and for any $S_i,T_j$, either $S_i \subseteq T_j$ or $S_i \cap T_j = \emptyset$.

\begin{definition}\label{def:top-level-graph}
 Given a (UI, CI)-pattern $p = (p_{=}, p_{\sim})$,
 the \emph{top-level bipartite graph} $G_{p}$ is defined as follows.
 Let the partite sets of $G_p$ be  ${\cal T}$ and $\simmed{U}{\sim}.$
 For each $T_r \in {\cal T}$ and class $\simmed{u}{\sim}$, we have an edge between $T_r$ and $\simmed{u}{\sim}$ if and only if there exists an authorized partial plan $\pi_r: T_r \rightarrow \simmed{u}{\sim}$ such that ${p_{=}}|_{T_r}$ is a UI-pattern for $\pi_r$.
\end{definition}

\begin{lemma}
  If a (UI, CI)-pattern $p = (p_{=}, p_{\sim})$ is realizable, then $G_p$ has a matching covering ${\cal T}$.   
\end{lemma}
\begin{proof}
 Let $\pi$ be an authorized plan such that $p$ is a (UI, CI)-pattern for $\pi$.
 As $p_{\sim}$ is a CI-pattern for $\pi$, we have that for each $T_r \in {\cal T}$ and all $s_i,s_j \in T_r$, $\simmed{\pi}{\sim}(s_i) = \simmed{\pi}{\sim}(s_j)$. 
 Therefore $\pi(T_r) \subseteq \simmed{u}{\sim}$ for some $u \in U$.
 Let $\simmed{u_r}{\sim}$ be this equivalence class for each $T_r$.
 As $p_{\sim}$ is a CI-pattern for $\pi$, we have that for all $r \neq r'$ and any $s_i \in T_r, s_j \in T_{r'}$, $\simmed{\pi}{\sim}(s_i) \neq \simmed{\pi}{\sim}(s_j)$. It follows that $\simmed{u_r}{\sim} \neq \simmed{u_{r'}}{\sim}$ for any $r \neq r'$.
 
Let $M = \{T_r\simmed{u_r}{\sim}\in E(G_p):\ T_r \in {\cal T}\}$. As $\simmed{u_r}{\sim} \neq \simmed{u_{r'}}{\sim}$ for any $r \neq r'$ we have that $M$ is a matching that covers  ${\cal T}$.
It remains to show that $M$ is a matching of $G_p$ covering ${\cal T}$, i.e. that $T_r\simmed{u_r}{\sim}$ is an edge in $G_p$ for each $T_r$.
For each $T_r \in {\cal T}$, let $\pi_r$ be $\pi$ restricted to $T_r$.
Then $\pi_r$ is a function from $T_r$ to $\simmed{u_r}{\sim}$.
As $\pi$ is authorized, $\pi_r$ is also authorized.
As $p_{=}$ is a UI-pattern for $\pi$,  we have that ${p_{=}}|_{T_r}$ is a UI-pattern for $\pi_r$.
Therefore $\pi_r$ satisfies all the conditions for there to be an edge $T_r\simmed{u_r}{\sim}$ in $G_p$.
\end{proof}

We have shown that for any (UI, CI)-pattern to be realizable, it must be consistent and its top-level bipartite graph must have a matching covering ${\cal T}$.
We will now show that these necessary conditions are also sufficient.

\begin{lemma}\label{lem6}
 Let  $p = (p_{=} = (x_1, \dots, x_k), p_{\sim} = (y_1, \dots, y_k))$ 
 be a (UI, CI)-pattern which is consistent, and such that $G_p$ has a matching covering ${\cal T}$.
 Then $p$ is realizable.
\end{lemma}
\begin{proof}
 Fix a matching $M$ in $G_p$ covering ${\cal T}$. 
 For each $T_r \in {\cal T}$, let $\simmed{u_r}{\sim} \in \simmed{U}{\sim}$ be the equivalence class of $U$ for which $T_r\simmed{u_r}{\sim}$ is an edge in $M$.
 Let $\pi_r$ be the authorized partial plan $\pi_r: T_r \rightarrow \simmed{u_r}{\sim}$ such that 
 ${p_{=}}|_{T_r}$ is a UI-pattern for $\pi_r$ (which must exist as  $T_r\simmed{u_r}{\sim}$ is an edge in $G_p$).
 Let $\pi = \bigcup_{T_r \in {\cal T}}\pi_r$.
 As each $\pi_r$ is authorized, $\pi$ is also authorized.
 It remains to show that $p$ is a (UI, CI)-pattern for $\pi$.
 
 We first show that $p_{\sim}$ is a CI-pattern for $\pi$. 
 Consider $y_i,y_j$ for any $i,j \in [k]$.
 If $y_i=y_j$, then $s_i,s_j \in T_r$ for some $r$, so by construction $\pi(s_i), \pi(s_j) \in \simmed{u_r}{\sim}$, and hence $\simmed{\pi}{\sim}(s_i) = \simmed{\pi}{\sim}(s_j)$.
 If $y_i\neq y_j$ then $\pi(s_i) \in \simmed{u_r}{\sim}$ and $\pi(s_j) \in \simmed{u_{r'}}{\sim}$, 
 and as  $M$ is a matching, $\simmed{u_r}{\sim} \neq \simmed{u_{r'}}{\sim}$.
 Therefore $\simmed{\pi}{\sim}(s_i) \neq \simmed{\pi}{\sim}(s_j)$.
 We therefore have that $p_{\sim}$ is a CI-pattern for $\pi$.
 
 We now show that $p_{=}$ is a UI-pattern for $\pi$. 
 Consider $x_i,x_j$ for any $i,j \in [k]$.
 If $x_i=x_j$, then as $p$ is consistent we also have $y_i=y_j$. Therefore $s_i,s_j \in T_r$ for some $r$. 
 As $\pi_r$ satisfies the conditions of the edge $T_r\simmed{u_r}{\sim}$, we have that $\pi_r(s_i)=\pi_r(s_j)$ and so $\pi(s_i)=\pi(s_j)$.
 If $x_i \neq x_j$, there are two cases to consider. 
If $y_i = y_j$, then again $s_i,s_j \in T_r$, and as $\pi_r$ satisfies the conditions of the edge $T_r\simmed{u_r}{\sim}$, $\pi_y(s_i) \neq \pi_y(s_j)$ and so $\pi(s_i) \neq \pi(s_j)$.
If on the other hand $y_i \neq y_j$, then by construction $\pi(s_i) \in \simmed{u_r}{\sim}$ and $\pi(s_j) \in \simmed{u_{r'}}{\sim}$ for some $r \neq r'$, and so $\pi(s_i) \neq \pi(s_j)$.
Thus $\pi_{=}$ is a UI-pattern for $\pi$, as required.
\end{proof}

\noindent{\bf The Second-level Bipartite Graph.}
For each (UI, CI)-pattern $p=(p_{=}, p_{\sim})$, we need to construct the graph $G_p$ and decide whether it has a matching covering ${\cal T}$, in order to decide whether $p$ is realizable. Given $G_p$, a maximum matching can be found in polynomial time using standard techniques, but constructing $G_p$ itself is non-trivial.
For each potential edge $T_r\simmed{u}{\sim}$ in $G_p$, we need to decide whether there exists an authorized partial plan $\pi_r: T_r \rightarrow \simmed{u}{\sim}$ such that ${p_{=}}|_{T_r}$ is a UI-pattern for $\pi_r$.
We can decide this 
by constructing another bipartite graph, $G_{T_r\simmed{u}{\sim}}$. Recall that ${\cal S} = \{S_1, \dots, S_l\}$ is a partition of $S$ into equivalence classes, where $s_i,s_j$ are equivalent if $x_i=x_j$, and for each $S_h \in {\cal S}$, either $S_h \subseteq T_r$ or $S_h \cap T_r = \emptyset$.
 Define ${\cal S}_r = \{S_h: S_h \subseteq T_r\}$.

 \begin{definition}\label{def:second-level-graph}
 Given a (UI, CI)-pattern $p = (p_{=} = (x_1, \dots, x_k), p_{\sim} = (y_1, \dots, y_k))$, a set $T_r \in {\cal T}$ and equivalence class $\simmed{u}{\sim} \in \simmed{U}{\sim}$, the \emph{second-level bipartite graph} $G_{T_r\simmed{u}{\sim}}$ is defined as follows:
Let the partite sets of $G$ be ${\cal S}_r$ and $\simmed{u}{\sim}$ and
 for each $S_h \in {\cal S}_r$ and $v \in \simmed{u}{\sim}$, we have an edge between $S_h$ and $v$ if and only if $v$ is authorized for all steps in $S_h$.
 \end{definition}

 \begin{lemma}\label{lem:2level}
  Given $T_r \in {\cal T}$, $\simmed{u}{\sim} \in \simmed{U}{\sim}$, the following conditions are equivalent.
  \begin{compactitem}
  \item There exists an authorized partial plan $\pi: T_r \rightarrow \simmed{u}{\sim}$ such that
  ${p_{=}}|_{T_r}$ is a UI-pattern for $\pi$.
  \item $G_{T_r\simmed{u}{\sim}}$ has a matching that covers ${\cal S}_r$.
  \end{compactitem}
 \end{lemma}
 \begin{proof}
 Suppose first that there exists an authorized partial plan $\pi: T_r \rightarrow \simmed{u}{\sim}$ such that ${p_{=}}|_{T_r}$ is a UI-pattern for $\pi$.
 For each $S_h \in {\cal S}_r$ and any $s_i,s_j \in S_h$, we have that $x_i=x_j$ and so $\pi(s_i)=\pi(s_j)$. So let $v_h$ be the user in $\simmed{u}{\sim}$ such that $\pi(s) = v_h$ for all $s \in S_h$.
 As $\pi$ is authorized, clearly $v_h$ is authorized for all $s \in S_h$, and so $S_hv_h$ is an edge in $G_{T_r\simmed{u}{\sim}}$. Furthermore for any $s_i \in S_h, s_j \in S_{h'}, h \neq h'$, we have that $x_i \neq x_j$ and so $\pi(s_i) \neq \pi(s_j)$ (as ${p_{=}}|_{T_r}$ is a UI-pattern for $\pi$). Therefore $M = \{S_hv_h: S_h \in {\cal S}_r\}$ is a matching in $G_{T_r\simmed{u}{\sim}}$ that covers ${\cal S}_r$, as required.
 
 Conversely, suppose that $G_{T_r\simmed{u}{\sim}}$ has a matching $M$ that covers ${\cal S}_r$.
 For each $S_h \in {\cal S}_r$, let $v_h$ be the user matched to $S_h$ in $M$.
 Let $\pi: T_r \rightarrow \simmed{u}{\sim}$ be the partial plan such that $\pi(s) = v_h \Leftrightarrow s \in S_h$.
 As $v_h \neq v_{h'}$ for any $S_h \neq S_{h'}$, and $x_i = x_j$ if and only if $s_i,s_j$ are in the same $S_h$, we have that $\pi(s_i)=\pi(s_j)$ if and only if $x_i = x_j$, and so ${p_{=}}|_{T_r}$ is a UI-pattern for $\pi$.
 Furthermore, as $v_h$ is authorized for all $s \in S_h$, $\pi$ is authorized, as required.
 \end{proof}
 
 \vspace{2mm}
 
 \section{FPT Algorithm}\label{sec:main}

\SetKwInOut{Input}{input}
\SetKwInOut{Output}{output}

\begin{algorithm}[!tb]
\caption{Main}\label{alg:start}
\Input {WSP instance $W = (S, U, \mathcal{A}, C)$}
\Output {UNSAT or SAT}
$p _{=} = 0^k$\;
$p_{\sim} = 0^k$\;
\Return{ PatBackTrack$(W,p_=,p_\sim)$};
\end{algorithm}


\begin{algorithm}[!h]
\caption{PatBackTrack$(W,p_=,p_\sim)$}\label{alg:recursion}
\Input {WSP instance $W = (S, U, \mathcal{A}, C)$, partial patterns $p_= = (x_1,\dots,x_k)$ and $p_{\sim} = (y_1,\dots,y_k)$}
\Output {UNSAT or SAT}
 \eIf{$p_=$ is complete and $p_\sim$ is complete}
 {
  \Return{Realizable$(W,p)$}\;
 }
 {
  \eIf{$p_=$ is incomplete}
  {
   Choose $i$ such that $x_i=0$\;
  \For{each $a \in \{1,\ldots,\max\{x_j : 1\le j\le k\} +1\}$}
   {  $x_i = a$\;
    \If{$\exists u$ authorized for all $s_j$ such that $x_j = a$ and $p_{=}$ is eligible}
    {
     \If{PatBackTrack$(W,p_=,p_\sim)$ returns SAT}{
      Return SAT\;
     }
    }
   }
  }
  {
    Choose $i$ such that $y_i = 0$\;
    \For{each $a \in \{1,\dots,\max\{y_j : 1 \le j \le k\} + 1\}$}
    {
     \For{each $j$ such that $x_j = x_i$}
      {
       $y_j = a$\;
      }
      \If{$p_\sim$ is eligible}
      {
       \If{PatBackTrack$(W,p_=,p_\sim)$ returns SAT}{
        Return SAT\;
       }
      }
    }
  }
 }
  Return UNSAT\;
\end{algorithm}

Algorithms \ref{alg:start} and \ref{alg:recursion} provide a partial pseudo-code of our FPT algorithm (still for the case of a single level of CI-constraints). To save space, we do not describe procedure {\em Realizable}($W, p$), which is a construction of bipartite graphs and search for matchings in those graphs as described in Section \ref{sec:real}.  We also omit a weight heuristic, which is described in Section~\ref{sec:experiments}.

 Our algorithm generates (UI, CI)-patterns $p$ in a backtracking manner as follows. 
It first generates partial patterns $p_{=} = (x_1, \dots, x_k)$, where
the coordinates $x_i=0$ are assigned one by one to integers in $[k']$,
where $k'=\max_{1\le j\le k} \{x_j\} +1$ (see Algorithm \ref{alg:recursion}). 
The algorithm keeps $x_i$ set to $a\in [k']$ only if there is a user authorized to perform all steps $s_j$ for which $x_j=a$ in $p_{=}$. The algorithm also checks that the pattern $p_{=}$ does not violate any constraints whose scope contain the corresponding step $s_i$. 

If an eligible pattern  $p_{=} = (x_1, \dots, x_k)$ has been completed (i.e., $x_j\neq 0$ for each $j\in [k]$), the partial patterns $p_{\sim} = (y_1, \dots, y_k)$ are generated as above but with two differences:
the algorithm ensures the consistency condition and no preliminary
authorizations checks are performed (see Algorithm \ref{alg:recursion}).

If an eligible (UI, CI)-pattern $p$ has been constructed, a procedure constructing bipartite graphs and searching for matchings in them as described in Section \ref{sec:real} decides whether $p$ is realizable. 
The algorithm stops when either a realizable and eligible pattern is found, or all eligible patterns have been considered and the WSP instance is declared unsatisfiable. 

The following theorem follows from the more general Theorem~\ref{thm:nested-algorithm}, given in Section~\ref{sec:nested}.
Note that the algorithm we describe above is equivalent to the special case of $r=2$ of the general algorithm.

\begin{theorem}\label{prop:complexity-realizability}
 We can solve WSP with UI and CI constraints in $O^*(2^{k\log (3k)})$ time.\footnote{In this paper, all logarithms are of base 2.}
\end{theorem}

%
%
%
%

\section{Nested Equivalence Relations}\label{sec:nested}

 Suppose we have a series of equivalence relations $\sim_1, \sim_2, \dots, \sim_r$, such that each equivalence relation is a refinement of the ones preceeding it, and a set of constraints $C_{\sim_q}$ for each equivalence relation $\sim_q$.
 Then we extend our approach as follows.\footnote{In reality, $r$ will be quite small and may be considered as a parameter alongside $k$.}
 For each equivalence relation $\sim_q$, we define a pattern $p_{\sim_q} = (x^q_1, \dots, x^q_r)$, where $x^q_i \in [k]$ for all $i \in [k]$.
 We will also write the pattern $p_{\sim_q}$ as $(p_{\sim_q}(1), \dots, p_{\sim_q}(r))$.
 We say that $p_{\sim_q}$ is a \emph{$\sim_q$-pattern} for a plan $\pi$ if $x^q_i=x^q_j \Leftrightarrow \pi(s_i) \sim_q \pi(s_j)$, for all $i,j \in [k]$.
 Given a plan $\pi$ and a $\sim_q$-pattern $p_{\sim_q}$ for $\pi$ for each $q \in [r]$, we define the tuple $p = (p_{\sim_1}, \dots, p_{\sim_r})$ to be a \emph{joint pattern} for $\pi$.
The algorithm now proceeds in a natural generalisation of the previously considered case where $q=2$ (see below),
but in order to analyse the running time more carefully we need to note some subtleties in the definitions of patterns and partitions.


\subsection{Joint Patterns and Nested Partitions}

Consider nested equivalence relations $\sim_1, \dots, \sim_r$ as above, 
and an instance of WSP with constraints $C_{\sim_1} \cup \dots \cup C_{\sim_r}$
where for each $i \in [r]$, $C_{\sim_i}$ contains constraints that are CI for $\sim_i$. 
We define a joint pattern $p=(p_{\sim_1}, \dots, p_{\sim_r})$ to be \emph{eligible} and \emph{realizable} 
in the natural way, extending the definitions used previously in this paper. 
Similarly, $p$ is \emph{consistent} if $x^{q+1}_i=x^{q+1}_j \Rightarrow x^q_i = x^q_j$ for all $q \in [r-1]$, $i, j \in [k]$.

As previously, it is easy to test in polynomial time whether a joint pattern is eligible,
and realizability can be tested via a generalisation of the approach described in Section~\ref{sec:real}; see below. 
Hence, the existence of an FPT algorithm (parameterized jointly by $k$ and $r$)
follows from the number of joint patterns being bounded, and from an algorithm  
for enumerating joint patterns (also given below). (The number of (not necessarily consistent) possible
joint patterns is clearly $k^{rk}$, which would also be the dominating term 
in a naive analysis of the running time.)

We briefly note that restricting our attention to consistent joint patterns does not improve this bound.
To see this, consider a plan $\pi$ where all steps are assigned to different $\sim_i$-equivalence classes
on all levels $i$, i.e., $\simmed{\pi}{\sim_1}(s) \neq \simmed{\pi}{\sim_1}(s')$ for all steps $s \neq s'$.
Then the number of consistent joint patterns corresponding to $\pi$ is exactly $(k!)^r=k^{\Theta(rk)}$,
since a different numbering scheme may be used at every level of the pattern. 

\newcommand{\cP}{\ensuremath{\mathcal P}}

For a better bound on the running time, we introduce some more terminology. 
Recall that a partition $\cP$ is a refinement of a partition $\cP'$ if they are partitions of the same ground set, and for every set $S \in \cP$ there is some set $S' \in \cP'$ such that $S \subseteq S'$.
A \emph{nested partition (in $r$ levels)} of a set $S$ is a tuple $\cP=(\cP_1, \ldots, \cP_r)$ of partitions of $S$ where $\cP_{i+1}$ is a refinement of $\cP_i$ for each $i \in [r-1]$. 
Thus a nested partition is essentially equivalent to a consistent joint pattern, except that no numbering for the partitions has been specified. We find that this has a significant impact on their number.

\begin{theorem} \label{thm:count-nested-partitions}
  Let $S$ be a set with $|S|=k$ and let $r$ be an integer. The number of nested partitions of $S$ in $r$ levels is at most $(r+1)^{k-1}k^{k-2}$.
\end{theorem}
\begin{proof}
  We will describe nested partitions in terms of edge-labelled trees, such that distinct nested partitions yield distinct edge-labelled trees; the bound will follow.

  We construct a tree $T$ on vertex set $S$ bottom-up as follows. To begin, let $T_r$ be an arbitrary forest corresponding to the partition $\cP_r$, i.e., the partition of $S$ into connected components of $T_r$ is exactly the partition $\cP_r$. Let every edge of $T_r$ have label $r$. Note that some components of $T_r$ may be edgeless, i.e., consist of only a single step $s \in S$. 

  Next, for each $i \in [r-1]$ we iteratively define a forest $T_i$ from $T_{i+1}$ by adding new edges to $T_{i+1}$, with label $i$, until $T_i$ corresponds to the partition $\cP_i$ (in the same sense as previously). Note that this is possible since $\cP$ is a nested partition. Again, the precise choice of edges is arbitrary (subject to these specifications). 

  Finally, we complete $T_1$ into a tree $T$ by adding edges with label $0$. This yields a tree over $S$ with edges labelled by $r+1$ different labels. By Cayley's formula, there are $k^{k-2}$ distinct trees on $S$, and for each tree there are $(r+1)^{k-1}$ different edge labellings; hence the number of distinct labelled trees matches the claimed bound.

It only remains to show that distinct nested partitions yield distinct labelled trees. This follows since the nested partition can be recovered from the labelled tree: by construction, the partition $\cP_i$ corresponds to the forest containing all edges of $T$ with label $j \geq i$. 

The result follows.
\end{proof}

In the rest of this section, we show how to give an FPT algorithm which enumerates distinct nested partitions.
(This will be very similar to the results of Sections~\ref{sec:pp}--\ref{sec:main}; indeed, it is not difficult
to see that the enumeration strategy shown in Algorithm~\ref{alg:recursion} meets this requirement.)

The discussion will focus on consistent joint patterns, since this notion matches the design of the algorithm more closely; 
we will return to the notion of nested partitions when we provide the running time bound. 


\subsection{Checking Realizability}

Let us discuss how to check realizability of a joint pattern.
Note that if $p$ is the joint pattern for a plan, then necessarily $p$ is consistent;
hence we assume that $p  = (p_{\sim_1}, \dots, p_{\sim_r})$ is a consistent pattern. 
 
 Rather than defining two layers of bipartite graphs in order to check realizability, we define $r$ layers. 
 For notational convenience, let $\sim_0$ be the trivial equivalence relation for which all users are in the same class, and let $p_{\sim_0}$ be a pattern matching every task to the same label.
 We assume that $\sim_r$ is the relation $=$ (if $\sim r$ is not the relation $=$, we need to introduce $=$ as a new relation $\sim_{r+1}$ and proceed with a larger value of $r$).

 For any $q \in \{0, \dots, r\}$, and any label $x$ appearing in $p_{\sim_q}$, let $S^q_x = \{s_i \in S: x^q_i = x\}$. 
 For $q  <r$, let ${\cal S}^q_x$ be the set of all $S^{q+1}_y$ for which $S^{q+1}_y \subseteq S^q_x$.
 (Note that as $p$ is consistent, for any labels $x,y$, either $S^{q+1}_y \subseteq S^q_x$ or $S^{q+1}_y \cap S^q_x = \emptyset$.)
 Let $\simmed{u}{\sim_q}$ be an equivalence class with respect to $\sim_q$.
 For any such equivalence class, $\simmed{(\simmed{u}{\sim_q})}{\sim_{q+1}}$ denotes the set of all equivalence classes of $\simmed{u}{\sim_q}$ with respect to $\sim_{q+1}$, i.e. the set of all classes $\simmed{v}{\sim_{q+1}}$ such that $\simmed{v}{\sim_{q+1}} \subseteq \simmed{u}{\sim_q}$.
 Then we define a \emph{$q$th-level bipartite graph} as follows:

 \begin{definition}
  Given a joint pattern $p = (p_{\sim_1} = ((x^1_1, \dots, x^1_k), \dots, p_{\sim_r} = (x^r_1, \dots, x^r_k))$, an integer $q \in [r]$, a set $S^{q-1}_x = \{s_i \in S: x^{q-1}_i = x\}$ and an equivalence class $\simmed{u}{\sim_{q-1}} \in \simmed{U}{\sim_{q-1}}$, the 
  \emph{$q$th-level bipartite graph $G_{S^{q-1}_x\simmed{u}{\sim_{q-1}}}$} is defined as follows:
  Let the vertex set of $G_{S^{q-1}_x\simmed{u}{\sim_{q-1}}}$ be ${\cal S}^{q-1}_x \cup \simmed{(\simmed{u}{\sim_{q-1}})}{\sim_{q}}$.
  For each $S^q_y \in {\cal S}^{q-1}_x, \simmed{v}{\sim_q} \in \simmed{(\simmed{u}{\sim_{q-1}})}{\sim_{q}}$, 
  we have an edge between $S^q_y$ and $\simmed{v}{\sim_q}$ if and only if 
  there exists an authorized partial plan $\pi_{q,y}: S^q_y \rightarrow \simmed{v}{\sim_q}$ such that $p_{q'}|_{S^q_y}$ is a $\sim_{q'}$-pattern for $\pi_{q,y}$ for each $q' \ge q$.  
 \end{definition}

 Similarly to previous lemmas, we can prove the following result:
 
 \begin{lemma}\label{lem:generalizedGraph}
  The following conditions are equivalent:
  (i) There exists an authorized partial plan $\pi_{q,y}: S^q_y \rightarrow \simmed{v}{\sim_q}$ such that $p_{q'}|_{S^q_y}$ is a $\sim_{q'}$-pattern for $\pi_{q,y}$ for each $q' \ge q$; and (ii)
    $G_{S^q_y\simmed{v}{\sim_q}}$ has a matching covering ${\cal S}^q_y$.
\end{lemma}

 Observe that (assuming $\sim_r$ is the relation $=$) if $q = r$, then  $\simmed{v}{\sim_q} = \{v\}$ and there is an edge between $S^q_y$ and $\simmed{v}{\sim_q}$ if and only if $v$ is authorized for all steps in $S^q_y$.
 Therefore $G_{S^{r-1}_x\simmed{u}{\sim_{r-1}}}$ can be constructed in polynomial time, and a matching saturating ${\cal S}^{r-1}_x$ can be found in polynomial time if one exists.
 By Lemma \ref{lem:generalizedGraph}, we can use graphs of the form $G_{S_x\simmed{u}{\sim_q}}$ to construct graphs of the form $G_{S_{y}\simmed{v}{\sim_{q-1}}}$.
 Thus, in polynomial time (for fixed $r$) we can decide whether  
 there exists an authorized partial plan $\pi_{0,y}: S^0_y \rightarrow \simmed{v}{\sim_0}$ such that $p_{q'}|_{S^0_y}$ is a $\sim_{q'}$-pattern for $\pi_{0,y}$ for each $q' \ge 0$.
 As $S^0_y = S$ and $\simmed{v}{\sim_0} = U$, this lets us decide whether there exists a complete, valid plan $\pi$ corresponding to an eligible joint pattern $p$.


\subsection{The Algorithm and Running Time}

The algorithm can now be constructed very similarly as in Section~\ref{sec:main}. 
We begin by defining an empty partial joint pattern $p=(0^k, \dots, 0^k)$, then 
as in Algorithm~\ref{alg:recursion} we construct a recursive backtracking algorithm
to complete $p$ into a complete joint pattern (where no entry is $0$). 

This is done in a bottom-up manner. Let $p'=(p_{\sim_1}', \ldots, p_{\sim_r}')$ be a partial joint pattern.
If $p'$ is complete, then we proceed to test realizability as above. Otherwise, 
let $i \leq r$ be the largest integer such that $p_{\sim_i}'$ is incomplete,
and let $j \in [k]$ be such that $p_{\sim_i}'(j)=0$. Let $k'=\max_{j' \in [k]} p_{\sim_i}'(j')+1$,
and let $S_{i,j}=\{j' \in [k]: p'_{\sim_i}(j') = p'_{\sim_i}(j)\}$. (If $i=r$, then we simply
define $S_{i,j}=\{j\}$.) Then for every $a \in [k']$ we perform the following procedure: 
fix $p'_{\sim_i}(j')=a$ for every $j' \in S_{i,j}$; check if the resulting partial pattern $p'_{\sim_i}$ 
is ineligible (i.e., if some constraint of $C_i$ whose scope intersects $S_{i,j}$ has become violated); 
and if not, make a recursive call with the resulting joint pattern $p'$.

We claim that this is a correct algorithm, which enumerates joint patterns which are consistent by the 
specification of the set $S_{i,j}$, and which furthermore enumerates only distinct nested partitions
thanks to the choice of $k'$. 

\begin{theorem} \label{thm:nested-algorithm}
  WSP with nested class-independent constraints in $r$ levels and with $k$ steps 
  is FPT with a running time of 
  $O^*(2^{k \log((r+1)k)})$.
\end{theorem}
\begin{proof}
  Clearly, since eligibility and authorization of every proposed joint pattern
  is verified explicitly, the algorithm gives no false positives, i.e., it never
  reports the existence of a valid plan for an unsatisfiable instance. 
  The opposite also holds: Assume that the instance allows for a valid plan $\pi$. 
  Then at every recursion point, corresponding to the specification of a value $p'_{\sim_i}(j)$, 
  there is exactly one value of $a$ consistent with $\pi$ (either $k'=1$ in which case
  there is no choice; or $p_{\sim_i}$ places $s_j$ in the same equivalence class
  as some previously specified step $s_j'$; or $s_j$ must be placed in a new equivalence class 
  and we let $a=k'$). It is also clear that this recursive path is not aborted. 
  Hence the process results in a complete joint pattern $p$ corresponding to $\pi$,
  which is eligible by assumption, and for which some authorized complete plan $\pi'$
  is subsequently computed. 

  To bound the running time, we argue very similarly to show that the number of leaves
  is bounded by the number of distinct nested partitions. Clearly, for an upper bound
  on the running time we may assume that no recursive branch is aborted (i.e.,
  every possible plan is eligible). Then we find as above that for every 
  nested partition $\cP$, we can trace exactly one path from the root of the calling tree
  to a leaf, where at every point there is exactly one value $p'_{\sim_i}(j)=a$
  that is consistent with $\cP$. We also find that every leaf of the calling tree,
  corresponding to a complete joint pattern $p$, corresponds to only exactly one 
  nested partition. Hence Theorem~\ref{thm:count-nested-partitions} bounds the number
  of leaves of the calling tree by $(r+1)^{k-1}k^{k-2} = 2^{(k-1)\log(r+1) + (k-1)\log k} = O^*(2^{k \log((r+1)k)})$. The total running time
  is bounded by a polynomial factor times this number; hence the result follows.
\end{proof}

Finally, we note that for modest values of $r$, specifically $r=k^{o(1)}$, then this bound
can be written as $O^*(2^{(k \log k)(1+o(1))})$, i.e., the overhead due to $r$ is not ``visible''
in the exponent until $r=k^{\Omega(1)}$.

Recently, it was shown that (under a standard complexity assumption)
even the basic case of WSP with only UI constraints admits no algorithm 
with a running time of $O^*(2^{(1-\varepsilon)k \log k})$ for any $\varepsilon>0$~\cite{GutinW15WSParXiv}. 
Hence for cases where $r=k^{o(1)}$, the bound in Theorem~\ref{thm:nested-algorithm}
matches the lower bound (up to lower-order terms in the exponent).

\section{Algorithm Implementation and Computational Experiments}\label{sec:experiments}

There can be a huge difference between an algorithm in principle and its actual implementation as a computer code. For example, see \cite{ChKl10,MyKo11}.
We have implemented the new pattern-backtracking FPT algorithm and a reduction to the pseudo-Boolean satisfiability (PB SAT) problem in C++, using SAT4J \cite{BePa10} as a pseudo-Boolean SAT solver.
Reductions from WSP constraints to PB ones were done similarly to those in \cite{CoCrGaGuJo14c,JOCO2014,KaGaGu}. Our FPT algorithm extends the pattern-backtracking framework of \cite{KaGaGu} in a nontrivial way; see below. 

In this section we first describe some tweaks and heuristics used by the algorithm (with no known impact on its theoretical performance), then we describe a series of experiments that we ran to test the performance of our FPT algorithm against that of SAT4J.
Due to the difficulty of acquiring real-world workflow instances, we generate and use synthetic data to test our new FPT algorithm and reduction to the PB SAT problem (as in similar experimental studies~\cite{JOCO2014,KaGaGu,WaLi10}).
All our experiments use a MacBook Pro computer having a 2.6 GHz Intel Core i5 processor, 8 GB 1600 MHz DDR3 RAM and running Mac OS X 10.9.5.

We generate a number of random WSP instances using not-equals (i.e, constraints of the form $(s,s',\neq)$),  equivalence and non-equivalence constraints (i.e., constraints of the types $(s,s',\sim)$ and  $(s,s',\nsim)$), and at-most constraints.
An {\em at-most constraint} is a UI constraint that restricts the number of users that may be involved in the execution of a set of steps.
It is, therefore, a form of cardinality constraint and imposes a loose form of ``need-to-know'' constraint on the execution of a workflow instance, which can be important in certain business processes.
An at-most constraint may be represented as a tuple $(t,Q,\leqslant)$, where $Q \subseteq S$, $1\leqslant t \leqslant |Q|$, and is satisfied by any plan that allocates no more than $t$ users in total to the steps in $Q$.
In all our at-most constraints $t=3$ and $|Q|=5$ as in
\cite{JOCO2014,KaGaGu}.

\subsection{Further Implementation Details}

The FPT algorithm and pattern generation of \cite{JOCO2014} have to assume a fixed ordering $s_1,\ldots ,s_k$ of steps in $S$, whereas the pattern-backtracking framework we use allows us to consider the steps as arbitrarily permuted and to browse the search space of patterns more efficiently. Our algorithm uses a heuristic to decide which zero-valued coordinate $x_i$ (when $p_=$ is constructed) or $y_i$  (when $p_\sim$ is constructed) should be considered next. The heuristic simply chooses a zero-valued coordinate of maximum weight, but the way to compute weights of zero-valued coordinates depends on the type of constraints in the WSP instance. 

For the types of constraints used in our computational experiments, the weights are computed as follows: the weight of $x_i$ is the total number of steps involved in user-independent constraints containing $s_i$, and the weight of $y_i$ is the number of non-equivalence constraints $(s,s',\not\sim)$ containing $s_i$ plus ten times the number of equivalence constraints $(s,s',\sim)$ constraining $s_i$.
The intuition behind this is as follows. For user-independent constraints, a step involved in user-independent constraints containing the largest number of steps in total reduces the pattern search space more effectively. Similarly, for class-independence constraints, a step involved in a larger number of constraints reduces the search space more effectively, with equivalence constraints having a much stronger influence on the search space reduction. In other words, we choose a ``more constrained" step in each case first. 

The procedure {\em Realizable}($W, p$), used to test realizability by finding matchings covering one partite set of the bipartite graph, uses a modified version of the Hungarian algorithm and data structures from \cite{KoKr04} in combination with some simple speed-ups and Proposition~1 of \cite{KaGaGu}.

\subsection{Experimental Parameters and Instance Generation}

We summarize the parameters we use for our experiments in Table~\ref{tbl:experimental-parameters}.
Values of $k$, $n$ and $r$ were chosen that seemed appropriate for real-world workflow specifications.
The values of the other parameters were determined by preliminary experiments designed to identify ``challenging'' instances of WSP: that is, instances that were neither very lightly constrained nor very tightly constrained.
Informally, it is relatively easy to determine that lightly constrained instances are satisfiable and that tighly constrained instances are unsatisfiable.
Thus the instances we use in our experiments are (very approximately) equally likely to be satisfiable or unsatisfiable.
In particular, by varying the numbers of at-most constraints and constraints of the form $(s,s',\nsim)$, we are able to generate a set of instances with the desired characteristics (as shown by the results in Table~\ref{tbl:big}).

\begin{table}[!t]\centering 

\caption{Parameters used in our experiments}\label{tbl:experimental-parameters}
 \begin{tabular}{|l|r|r|}
 \hline
  \multicolumn{2}{|l|}{\bf Parameter} & \bf Values \\
 \hline
  \multicolumn{2}{|l|}{Number of steps $k$} & $20,25,30$ \\
 \hline
  \multicolumn{2}{|l|}{Number of users $n$} & $10k$ \\
 \hline
  \multicolumn{2}{|l|}{Number of user equivalence classes $r$} & $2k$ \\
 \hline
    & $k=20$ & $20,25$ \\ 
  Number of constraints  $(s,s',\ne)$ & $k=25$ & $25,30$ \\ 
    & $k=30$ & $30,35$ \\
 \hline
    & $k = 20$ & $0$ \\
  Number of constraints $(s,s',\sim)$ & $k = 25$ & $1$ \\
    & $k = 30$ & $2$ \\
 \hline
    & $k = 20$ & $10,15,20,25,30$ \\
  Number of constraints $(s,s',\nsim)$ & $k = 25$ & $15,20,25,30,35$ \\
    & $k = 30$ & $20,25,30,35,40$ \\
 \hline
    & $k = 20$ & $10,15,20,25,30,35,40$ \\
  Number of at-most constraints & $k = 25$ & $15,20,25,30,35,40,45$ \\
    & $k = 30$ & $20,25,30,35,40,45,50$ \\
 \hline
 \end{tabular}
\end{table}

A constraint $(s,s',\nsim)$ implies the existence of a constraint $(s,s',\ne)$, so we do not vary the number of not-equals a great deal (in contrast to existing work in the literature~\cite{JOCO2014}).
Informally, a constraint $(s,s',\sim)$ reduces the difficulty of finding a valid plan.
Thus, given our desire to investigate challenging instances, we do not use very many of these constraints.

All the constraints, authorizations, and equivalence classes of users are generated for each instance separately, uniformly at random. 
The random generation of authorizations, not-equals, and at-most constraints uses existing techniques~\cite{JOCO2014}. 
The generation of equivalence and non-equivalence constraints has to be controlled to ensure that an instance is not trivially unsatisfiable.
In particular, we must discard a constraint of the form $(s,s',\nsim)$ if we have already generated a constraint of the form $(s,s',\sim)$.
The equivalence classes of the user set are generated by enumerating the user set and then splitting the list into contiguous sublists.
The number of elements in each sublist varies between $3$ and $7$ (chosen uniformly at random and adjusted, where necessary, so that the total number of members in the $r$ sub-lists is $n$).

\subsection{Results and Evaluation}

We adopt the following labelling convention for our test instances: $a.b.c.d$ denotes an instance with $a$ not-equals constraints, $b$ at-most constraints, $c$ equivalence constraints, and $d$ non-equivalence constraints (as used in the first and fourth columns of Table~\ref{tbl:big}, for instances with $k=25$ and $k=30$, respectively).
In our experiments we compare the run-times and outcomes of SAT4J (having reduced the WSP instance to a PB SAT problem instance) and our FPT algorithm, which we will call PBA4CI (pattern-based algorithm for class-independent constraints).
Table~\ref{tbl:big} shows some detailed results of our experiments (the results for $k=20$
were excluded for reasons of space).
We record whether an instance is solved, indicating a satisfiable instance with a `Y' and an unsatisfiable instance with a `N'; instances that were not solved are indicated by a question mark.
PBA4CI reaches a conclusive decision (Y or N) for every test instance, whereas SAT4J fails to reach such a decision for some instances, typically because the machine runs out of memory.
The table also records the time (in seconds) taken for the algorithms to run on each instance.
We would expect that the time taken to solve an instance would depend on whether the instance is satisfiable or not, and this is confirmed by the results in the table.

\begin{table}[!t]\centering\setlength{\tabcolsep}{3pt}\setlength{\extrarowheight}{-1pt}
\caption{Results for $k = 25$ and $30$. Time in seconds. Y,N,? mean satisfied, unsatisfied, unsolved.} \label{tbl:big}
\begin{tabular}{|c|rr|rr||c|rr|rr|}
\hline
\multicolumn{1}{|c|}{\bf Instance} & \multicolumn{2}{c|}{\bf SAT4J} & \multicolumn{2}{c||}{\bf PBA4CI} & \multicolumn{1}{c|}{\bf Instance} & \multicolumn{2}{c|}{\bf SAT4J} & \multicolumn{2}{c|}{\bf PBA4CI} \\
\hline
\multicolumn{5}{|c||}{$k = 25$} & \multicolumn{5}{c|}{$k = 30$} \\
\hline
25.15.1.15 & Y & 2.62 & Y & 2.464 &30.20.2.20 & Y & 2.72 & Y & 50.804\\
25.20.1.15 & Y & 22.38 & Y & 0.010 &30.25.2.20 & Y & 271.78 & Y & 2.323\\
25.25.1.15 & Y & 11.03 & Y & 0.010 &30.30.2.20 & ? & 2,141.60 & Y & 2.946\\
25.30.1.15 & Y & 35.54 & Y & 0.040 &30.35.2.20 & ? & 2,250.02 & N & 0.412\\
25.35.1.15 & N & 1,439.94 & N & 0.075 &30.40.2.20 & ? & 1,942.57 & N & 2.238\\
25.40.1.15 & ? & 2,088.06 & N & 0.033 &30.45.2.20 & ? & 2,198.02 & N & 2.171\\
25.45.1.15 & Y & 113.37 & Y & 0.022 &30.50.2.20 & ? & 2,580.81 & N & 0.494\\
\hline
25.15.1.20 & Y & 1.52 & Y & 111.799 &30.20.2.25 & Y & 4.18 & Y & 237.604\\
25.20.1.20 & Y & 7.77 & Y & 0.024 &30.25.2.25 & Y & 76.41 & Y & 0.789\\
25.25.1.20 & Y & 297.39 & Y & 0.065 &30.30.2.25 & ? & 2,288.07 & N & 0.401\\
25.30.1.20 & ? & 2,273.56 & N & 0.033 &30.35.2.25 & Y & 1,364.66 & Y & 0.238\\
25.35.1.20 & Y & 48.29 & Y & 0.067 &30.40.2.25 & ? & 2,383.92 & N & 0.775\\
25.40.1.20 & N & 105.48 & N & 0.045 &30.45.2.25 & ? & 1,743.87 & N & 0.394\\
25.45.1.20 & ? & 2,105.61 & N & 0.031 &30.50.2.25 & ? & 2,385.39 & N & 0.218\\ 
\hline
25.15.1.25 & Y & 14.40 & Y & 0.014 &30.20.2.30 & Y & 35.40 & Y & 0.071\\
25.20.1.25 & Y & 80.25 & Y & 0.021 &30.25.2.30 & Y & 9.37 & Y & 1.063\\
25.25.1.25 & ? & 2,284.78 & N & 0.023 &30.30.2.30 & N & 1,632.51 & N & 0.347\\
25.30.1.25 & N & 442.91 & N & 0.237 &30.35.2.30 & Y & 803.50 & Y & 0.029\\
25.35.1.25 & ? & 2,188.01 & N & 0.060 &30.40.2.30 & ? & 2,022.71 & N & 0.981\\
25.40.1.25 & ? & 2,293.77 & N & 0.043 &30.45.2.30 & ? & 1,902.84 & N & 1.501\\
25.45.1.25 & ? & 2,041.02 & N & 0.144 & 30.50.2.30 & ? & 1,730.93 & N & 0.467\\
\hline
25.15.1.30 & Y & 3.22 & Y & 0.011 & 30.20.2.35 & Y & 24.12 & Y & 0.453\\
25.20.1.30 & Y & 240.59 & Y & 0.014 & 30.25.2.35 & Y & 456.51 & Y & 0.085\\
25.25.1.30 & Y & 66.74 & Y & 0.050 & 30.30.2.35 & N & 1,817.76 & N & 1.088\\
25.30.1.30 & ? & 2,301.75 & N & 0.088 & 30.35.2.35 & ? & 1,949.77 & N & 0.111\\
25.35.1.30 & N & 1,562.30 & N & 0.023 & 30.40.2.35 & ? & 2,115.32 & N & 0.551\\
25.40.1.30 & ? & 2,332.07 & N & 0.127 & 30.45.2.35 & ? & 1,535.57 & N & 0.118\\
25.45.1.30 & N & 950.25 & N & 0.040 & 30.50.2.35 & ? & 1,647.41 & N & 0.454\\
\hline
25.15.1.35 & Y & 10.57 & Y & 0.014 & 30.20.2.40 & ? & 3,088.54 & N & 0.729\\
25.20.1.35 & N & 218.70 & N & 0.166 & 30.25.2.40 & ? & 1,746.81 & Y & 0.542\\
25.25.1.35 & Y & 37.87 & Y & 0.012 & 30.30.2.40 & ? & 2,350.01 & Y & 0.949\\
25.30.1.35 & ? & 2,421.30 & N & 0.054 & 30.35.2.40 & ? & 1,857.27 & N & 0.576\\
25.35.1.35 & N & 1,524.68 & N & 0.022 & 30.40.2.40 & ? & 1,938.63 & N & 0.221\\
25.40.1.35 & N & 1,001.67 & N & 0.028 & 30.45.2.40 & ? & 2,159.50 & N & 0.209\\
25.45.1.35 & ? & 1,974.05 & N & 0.034 & 30.50.2.40 & ? & 1,815.15 & N & 0.337\\
\hline
\end{tabular}
\end{table}

In total, the experiments cover $210$ randomly generated instances, $70$ instances for each number of steps, $k\in\{20,25,30\}$. PBA4CI successfully solves all of the instances, while SAT4J fails on almost 40\% of the instances (mostly unsatisfiable ones). In terms of CPU time, SAT4J is more efficient only on $5$ instances ($2.4\%$) in total: $1$ for $20$ steps, $1$ for $25$ steps, and $3$ for $30$ steps, all of which are lightly constrained. 
For these instances PBA4CI has to generate a large number of patterns in the search space before it finds a solution. 

\begin{table}[!tb]\centering\setlength{\tabcolsep}{3pt}\setlength{\extrarowheight}{-1pt}
\caption{Summary statistics for $k \in \set{20,25,30}$}\label{tbl:overall}
\begin{tabular}{|r|r|r|r|r|r|}
\cline{3-6}
  \multicolumn{2}{c}{} & \multicolumn{2}{|c|}{\bf SAT4J} & \multicolumn{2}{c|}{\bf PBA4CI} \\
\hline
{$k$} & {Result} & {Count} & {Mean Time} & {Count} & {Mean Time} \\
\hline
20 & Y & 32 & 27.25 & 32 & 0.11  \\
   & N & 29\,(9) & 1,538.65 & 38 & 0.01 \\
     \cline{3-6}
   & Total & 61\,(9) & 847.72 & 70 & 0.06\\
\hline
\hline
25 & Y & 28 & 61.86 & 28 & 4.12\\
   & N & 15\,(27) & 1,719.31 & 42 & 0.07\\
     \cline{3-6}
   & Total & 43\,(27) & 1,056.33 & 70 & 1.69 \\      
\hline
\hline
30 & Y & 18\,(4) & 693.53 & 22  & 14.80\\
   & N & 6\,(42) & 2,003.76 & 48 & 0.84\\
     \cline{3-6}
   & Total & 24\,(46) & 1,591.97 & 70 & 5.23 \\      
\hline
\end{tabular}
\end{table}

Overall, PBA4CI is clearly more effective and efficient than SAT4J on these instances. 
Table~\ref{tbl:overall} shows the summary statistics for all the experiments. 
The numbers of unsolved instances by SAT4J are indicated in parenthesis. 
For average CPU time values, we assume that the running time on the unsolved instances can be considered as a lower bound on the time required to solve them. Therefore average time values in Table~\ref{tbl:overall} take into consideration unsolved instances for SAT4J: they are estimated lower bounds on its average time performance.
As the number of steps $k$ increases, SAT4J fails more frequently and is unable to reach a conclusive decision for more than $65\%$ of instances when $k = 30$, some of which are satisfiable.
However, SAT4J is clearly more efficient (and effective) on satisfiable instances than on the unsatisfiable ones, while for PBA4CI the converse is true. This can be explained by very different search strategies used by the solvers.

\section{Conclusion}\label{sec:con}

We have introduced the concept of a class-independent constraint, which significantly generalizes user-independent constraints and substantially extends the range of real-world business requirements that can be modelled.  
We have designed an FPT algorithm for WSP with class-independent constraints.
Our computational results demonstrate that our FPT algorithm is useful in practice for WSP with class-independent constraints, in particular for WSP instances that are too hard for SAT4J.

We also outlined a generalization of our approach, and gave a more careful analysis of the worst-case complexity compared to the previous version of this paper~\cite{CrGaGuJo}.

\vspace{2mm}

\noindent{\bf Acknowledgement.}  This research was partially supported by an EPSRC grant EP/K005162/1. The FPT algorithm's executable code and experimental data set are publicly available~\cite{Gagarin2015}.

\end{document}